\documentclass[runningheads]{llncs}
\usepackage{amsmath,amssymb}
\usepackage{mathtools}
\usepackage{mathrsfs}
\usepackage{graphicx}
\usepackage{xcolor}
\newcommand{\AX}[1]{\textnormal{\textbf{#1}}}
\definecolor{orange}{rgb}{1,0.5,0}
\usepackage{enumitem}
\usepackage{ifthen}
\usepackage{underscore}
\spnewtheorem{fact}{Observation}{\bfseries}{\itshape}

\definecolor{orcidgreen}{RGB}{166,206,57}
\newcommand{\orcidicon}{\includegraphics[width=0.26cm]{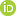}}
\def\orcidID#1{\renewcommand{\thefootnote}{\fnsymbol{footnote}}
	\unskip$^{\orcidicon}$\footnote{\orcidicon\color{orcidgreen}\,\footnotesize #1}\unskip
	\renewcommand{\thefootnote}{\arabic{footnote}}\unskip
}
\newcommand{\Hasse}[1][]{\mathfrak{H}\ifthenelse{\equal{#1}{}}{}{(#1)}}

\newcommand{\overlaps}{\between}
\DeclareMathOperator{\CC}{\mathtt{C}}
\DeclareMathOperator{\lca}{lca}
\DeclareMathOperator{\klca}{\textit{k}-lca}
\DeclareMathOperator{\2lca}{\textit{2}-lca}
\DeclareMathOperator{\LCA}{LCA}
\DeclareMathOperator{\Anc}{Anc}
\DeclareMathOperator{\cl}{cl}

\begin{document}
	
	\title{Unique Least Common Ancestors and Clusters in Directed Acyclic
		Graphs}
	
	\author{Ameera Vaheeda Shanavas\inst{1}\orcidID{0000-0003-0038-5492} \and
		Manoj Changat\inst{1}\orcidID{0000-0001-7257-6031} \and
		Marc Hellmuth\inst{2}\orcidID{0000-0002-1620-5508} \and 
		Peter F. Stadler\inst{2-6}\orcidID{0000-0002-5016-5191}
	}
	\authorrunning{A.\ Vaheeda Shanavas \emph{et al.}}
	\titlerunning{Last Common Ancestors and Clusters in DAGs}
	
	\institute{Department of Futures Studies, University of Kerala, Trivandrum
		695 581, India \email{ameerasv@gmail.com, mchangat@keralauniversity.ac.in}
		\and
		Dept.\ Mathematics, Faculty of Science,
		Stockholm University, SE-10691 Stockholm, Sweden.
		\email{marc.hellmuth@math.su.se},
		\and
		Bioinformatics Group, Department of Computer Science and Interdisciplinary Center for Bioinformatics, Universit{\"a}t Leipzig,
		H{\"a}rtelstrasse 16-18, D-04107 Leipzig, Germany,
		\email{studla@bioinf.uni-leipzig.de}
		\and
		Max Planck Institute for Mathematics in the Sciences, Leipzig, Germany \and
		Institute for Theoretical Chemistry,
		University of Vienna, Vienna, Austria \and
		Faculty of Ciencias, Universidad Nacional de Colombia, Bogot{\'a},
		Colombia
		\and 
		Santa Fe Institute, Santa Fe, NM }
	
	\maketitle

	\begin{abstract}
		We investigate the connections between clusters and least common
		ancestors (LCAs) in directed acyclic graphs (DAGs). We focus on the class
		of DAGs having unique least common ancestors for certain subsets of their
		minimal elements since these are of interest, particularly as models of
		phylogenetic networks. Here, we use the close connection between the
		canonical $k$-ary transit function and the closure function on a set
		system to show that pre-$k$-ary clustering systems are exactly those that
		derive from a class of DAGs with unique LCAs. Moreover, we show that $k$-ary $\mathscr{T}$-systems and $k$-weak hierarchies are associated 
		with DAGs that satisfy stronger conditions on
		the existence of unique LCAs for sets of size at most $k$.
		
	\end{abstract}
	
	\sloppy  
	
	\section{Introduction} 
	
	Directed acyclic graphs (DAGs) play an increasing role in mathematical
	phylogenetics as models of more complex evolutionary relationships that are
	not adequately represented by rooted trees. The set $X$ of minimal vertices
	of a DAG $G=(V,E)$ corresponds to the extant taxa and thus generalizes the
	leaf set of a phylogenetic tree. Inner vertices $u\in V$ are interpreted as
	ancestral states and are naturally associated with the sets $\CC(u)$ of
	the descendant genes. These sets are often called the ``hardwired
	clusters'' \cite{Nakhleh:05,Huson:11}. A least common ancestor of a set $A$
	of taxa is a minimal vertex $v$ in $G$ such that $A\subseteq\CC(v)$, i.e.,
	all taxa in $A$ are descendants of $v$. In phylogenetics, least common
	ancestors play a key role in understanding evolutionary relationships and
	processes. The clusters of $G$, on the other hand, are often accessible
	from data. Basic relations between clustering systems of (rooted) DAGs and
	the uniqueness of least common ancestors were explored recently in
	\cite{Hellmuth:22q}. Here, we elaborate further on this theme, making
	use in particular of the fact that the canonical transit function of set
	systems is a restriction of the closure function to small spanning sets.
	
	Section 2 contains basic definitions, some useful properties of DAGs, and a characterization of $k$-weak hierarchies.  Section 3 is about $\lca$ and $\klca$-property, and the connection of the latter with the pre-$k$-ary clustering systems. Section 4 discusses the correspondence of the strict and the strong $\klca$ properties to the $k$-ary $\mathscr{T}$-systems and the $k$-weak hierarchies, respectively.
	
	\section{Background and Preliminaries} 
	
	\paragraph{\textbf{Transit functions and $\boldsymbol{k}$-ary transit
			functions.}}  Let $X$ be a non-empty, finite set. We write $X^k$ for
	the $k$-fold Cartesian set product of $X$ and $X^{(k)}$ for the set of all
	non-empty subsets of $X$ with cardinality at most $k$.
	
	Following \cite{Changat:10}, a \emph{$k$-ary transit function} on $X$ is a
	function $R: X^k \mapsto 2^X$ satisfying the axioms
	\begin{description}[noitemsep,nolistsep] 
		\item[\AX{(t1)}] $u_1 \in R(u_1,u_2,\ldots, u_k)$;
		\item[\AX{(t2)}] $R(u_1,u_2,\ldots ,u_k)=R(\pi(u_1,u_2,\ldots, u_k))$ 
		for all $u_i\in X$ and all permutations $\pi$ of $(u_1,u_2,\ldots,u_k)$;
		\item[\AX{(t3)}] $R(u,u,\ldots,u)=\{u\}$ for all $u\in X$.
	\end{description}
	The ``symmetry'' axiom \AX{(t2)} allows us to interpret a $k$-ary transit
	function also as a function over subsets $U\in X^{(k)}$. Then axiom \AX{(t2)}
	becomes void, \AX{(t3)} becomes $R(\{x\})=\{x\}$ for all $x\in X$, and
	condition \AX{(t1)} reads ``$u\in U$ implies $u\in R(U)$ for all $U\in
	X^{(k)}$''.
	
	Given a $k$-ary transit function $R$ on $X$, we denote its system of
	\emph{transit sets} by $\mathscr{C}_R\coloneqq 
	\{R(U)\mid U\in X^{(k)}\}$. 
	A set system $\mathscr{C}\subset 2^X$ \emph{is identified} by a $k$-ary transit function if $\mathscr{C}=\mathscr{C}_{R_{\mathscr{C}}}$ where $R_{\mathscr{C}}: X^{(k)}\to 2^X$ defined by $R_{\mathscr{C}}(U) \coloneqq \bigcap\{ C\in\mathscr{C}\mid U\subseteq C\}$
	for all $U\in X^{(k)}$. As shown in \cite{Barthelemy:08,Changat:19a},
	a system of non-empty sets $\mathscr{C}\subset 2^X$ is identified by a
	$k$-ary transit function, $k\ge2$, if and only if $\mathscr{C}$ is a \emph{($k$-ary) $\mathscr{T}$-system}, satisfying the following three axiom
	\begin{description}[noitemsep,nolistsep] 
		\item[\AX{(KS)}] $\{x\}\in\mathscr{C}$ for all $x\in X$
		\item[\AX{(KR)}] For all $C\in\mathscr{C}$ there is a set $T\subseteq
		C$ with $|T|\le k$ such that $T\subseteq C'$ implies $C\subseteq C'$ for
		all $C'\in\mathscr{C}$. 
		\item[\AX{(KC)}] For every $U\subseteq X$ with $|U|\le k$ holds
		$\displaystyle\bigcap\{ C\in\mathscr{C}\mid U\subseteq C\} \in
		\mathscr{C}$.
	\end{description}
	Conversely, a $k$-ary transit function $R$
	identifies a set system if and only if it satisfies the \emph{monotone} axiom
	\begin{description}[noitemsep,nolistsep] 
		\item[\AX{(m)}] For every $w_1,\dots,w_k \in R(u_1,\dots,u_k)$ holds
		$R(w_1,\dots,w_k)\subseteq R(u_1,\dots,u_k)$.
	\end{description}
	That is, for all $U,W\in X^{(k)}$ holds: $W\subseteq R(U)$ implies $R(W)\subseteq R(U)$. 
	The correspondence of monotone $k$-ary transit functions and $k$-ary
	$\mathscr{T}$-systems is then mediated by the \emph{canonical transit function} $R_{\mathscr{C}}$ and $\mathscr{C}_R$, respectively.
	
	For a general set system $\mathscr{C}$ on $X$, the \emph{closure} function 
	$\cl:2^X \to 2^X$, sometimes also called the ``convex hull'', is defined as
	$\cl(A) \coloneqq \bigcap \{C \in\mathscr{C} \mid A\subseteq C \}$ for all
	$A\in 2^X$.
	The canonical $k$-ary transit
	function $R_{\mathscr{C}}$ of a set system is the restriction of its closure
	function to small sets as arguments: $R_{\mathscr{C}}(U)=\cl(U)$ for all
	non-empty sets $U$ with $|U|\le k$.
	
	A set system is \emph{closed} if for all non-empty set $A\in 2^X$ holds
	$A\in\mathscr{C}\iff \cl(A)=A$. By \cite[L.\ 16]{Hellmuth:22q}, this is
	equivalent to the condition that for all $A,B\in\mathscr{C}$ with $A\cap
	B\ne\emptyset$ we have $A\cap B\in \mathscr{C}$, i.e., $\mathscr{C}$ is
	closed under pairwise intersection.  A set system $\mathscr{C}$
	consisting of \emph{non-empty} subsets of $X$ is called a \emph{clustering
		system} if it satisfies \AX{(KS)} and \AX{(K1)}: $X\in\mathscr{C}$. Note
	that axiom \AX{(KS)} translates to $R_{\mathscr{C}}$ satisfying \AX{(t3)}.
	A $k$-ary $\mathscr{T}$-system is thus a clustering system if and only it
	satisfies \AX{(K1)} or, equivalently \cite{Changat:19a}, if its canonical
	transit function $R=R_{\mathscr{C}}$ satisfies
	\begin{description}[noitemsep,nolistsep] 
		\item[\AX{(a')}] there is $U\in X^{(k)}$ such that $R(U)=X$.
	\end{description}
	
	A set system is called \emph{pre-$k$-ary} if it satisfies \AX{(KC)} for a
	given parameter $k$. The $2$-ary case has received considerable attention
	in the literature for the special case of clustering systems, see \cite{Barthelemy:08}.  A $2$-ary transit function is called
	a \emph{transit function}. 
	A clustering system $\mathscr{C}$ is
	called pre-binary in \cite{Barthelemy:08} if \AX{(KC)} with $k=2$ is
	satisfied, i.e., if $R_{\mathscr{C}}(x,y)\in\mathscr{C}$ for all $x,y\in
	X$, and \emph{binary} if in addition \AX{(KR)} holds with $k=2$.  Binary
	clustering systems are therefore identified by monotone ($2$-ary) transit
	function satisfying \AX{(a')} with $k=2$; that is, there is $p,q\in X$ such
	that $R(p,q)=X$.
	
	\paragraph{\textbf{Weak and $\boldsymbol{k}$-weak hierarchies.}}
	Generalizations of hierarchies are important in the clustering
	literature. Recall that a clustering system $\mathscr{C}$ is a
	\begin{description}[noitemsep,nolistsep]
		\item[\emph{weak hierarchy}] if for any three sets $A,B,C\in\mathscr{C}$
		holds $A \cap B \cap C \in \{ A\cap B, A\cap C, B\cap C\}$
		\cite{Bandelt:89};
		\item[\emph{$\boldsymbol{k}$-weak hierarchy}] if for any $k+1$ sets
		$A_1,A_2,\dots,A_{k+1}\in\mathscr{C}$ there is $1\leq j\leq k+1$ such
		that $\displaystyle
		\bigcap_{i=1}^{k+1} A_i = \bigcap_{i=1,i\neq j}^{k+1} A_i$
		\cite{Bandelt:94}.
	\end{description}
	We write $A\overlaps B$ if $A\cap B\notin\{A,B,\emptyset\}$ and say that
	$A$ and $B$ overlap. It is well known that
	weak hierarchy = $2$-weak hierarchy $\implies$
	$k$-weak hierarchy $\implies$ $(k+1)$-weak hierarchy for all $k\ge 3$. As
	outlined in \cite{DRESS:96}, weak hierarchies always satisfy \AX{(KR)} for
	$k=2$. More generally, Lemma~6.3 of \cite{Changat:19a} ensures that
	$k$-weak hierarchies satisfy \AX{(KR)} for the parameter $k$.  For weak
	hierarchies, furthermore, axiom \AX{(KC)} with $k=2$ is equivalent to
	$\mathscr{C}$ being closed under pairwise intersection.
	
	The characterization of $k$-weak hierarchies by condition \AX{(kW')} in
	\cite{Bertrand:17}, together with the fact that every $k$-weak hierarchy is
	also a $k'$-weak hierarchy for all $k'\ge k$, can be rephrased as follows:
	\begin{fact}
		\label{obs:z}
		A set system $\mathscr{C}$ is a $k$-weak hierarchy if and only if for
		every $A\in 2^X$ with $|A|>k$ there is $z\in A$ such that
		$z\in\cl(A\setminus\{z\})$.
	\end{fact}
	For our purposes, the following characterization of $k$-weak hierarchies in
	terms of their closure functions will be particularly useful:
	\begin{proposition}
		\label{prop:kwH}
		A set system $\mathscr{C}$ on $X$ is a $k$-weak hierarchy if and only if
		for every $\emptyset \neq A\subseteq X$ there exists $U\subseteq A$
		with $|U|\leq k$ such that $\cl(A)=\cl(U)$.
	\end{proposition}
	\begin{proof}
		First, assume that $\mathscr{C}$ is a $k$-weak hierarchy. If $|A|\le k$,
		then $A=U$ trivially satisfies $\cl(A)=\cl(U)$.  Hence, assume $|A|>k$. By
		Obs.~\ref{obs:z}, there is $z\in A$ such that $z\in\cl(A\setminus\{z\})$,
		which implies $A\subseteq\cl(A\setminus\{z\})$. Together with isotony and
		idempotency of the closure function, we obtain
		\begin{equation*}
			\cl(A\setminus\{z\})\subseteq \cl(A)\subseteq
			\cl(\cl(A\setminus\{z\})) =
			\cl(A\setminus\{z\}).
		\end{equation*}
		Thus, there is $z\in A$ such that $\cl(A)=\cl(A\setminus\{z\})$.
		Repeating this argument for $A'\coloneqq A\setminus\{z\}$, we observe
		that we can stepwisely remove elements of $A$ while preserving $\cl(A)$
		until we arrive a residual set $U\subset A$ with $|U|=k$ that still
		satisfies $\cl(U)=\cl(A)$. 
		
		Now assume that $\mathscr{C}$ is \emph{not} a $k$-weak hierarchy.  Hence,
		there are $k+1$ sets $A_1,A_2,\dots,A_{k+1}\in\mathscr{C}$ such that for
		all $1\leq j\leq k+1$ it holds that $\cap_{i=1}^{k+1} A_i \subsetneq
		\cap_{i=1,i\neq j}^{k+1} A_i$.  Thus, there are $k+1$ (distinct) elements
		$x_1,\dots,x_{k+1}\in X$ such that $x_i\in A_j$ if and only $i\ne j$.
		Set $A=\{x_1,x_2,\dots,x_{k+1}\}$ and consider any subset $U\subset A$ with
		$|U|\leq k$.  Then there is at least one set $A_h$, $1\le h\le k+1$, such
		that $U\subseteq A_h$. By the previous arguments, $x_h\notin A_h$.  Since
		$A_h\in\mathscr{C}$, we have $\cl(U)\subseteq A_h$, and thus
		$x_h\notin\cl(U)$. Since $x_h\in A$ and $A\subseteq\cl(A)$, we have
		$x_h\in \cl(A)$ and, thus, $\cl(U)\ne\cl(A)$.  \qed
	\end{proof}
	
	\paragraph{\textbf{Clusters, $\boldsymbol{\LCA}$ and $\boldsymbol{\lca}$
			in DAGs.}}  Let $G$ be a directed acyclic graph (DAG) with an associated
	partial order $\preceq$ on its vertex set $V(G)$ defined by $v\preceq_G w$ if and only if there is a
	directed path from $w$ to $v$.  In this case, we say that $w$ is an
	ancestor of $v$ and $v$ is a descendant of $w$. If the context is clear, we
	may drop the subscript and write $\preceq$. Two vertices $u,v\in
	V(G)$ are \emph{incomparable} if neither $u\preceq v$ nor $v\preceq u$ is
	true.  We denote by $X=L(G)\subseteq V(G)$ the $\preceq$-minimal vertices
	of $G$ and we call $x\in X$ a leaf of $G$.  For every $v\in V(G)$, the set
	of its descendant leaves
	\begin{equation}
		\CC(v)\coloneqq\{ x\in X\mid x \preceq v\}
	\end{equation}
	is a cluster of $G$. We write $\mathscr{C}_G\coloneqq\{\CC(v)\mid v\in V(G)\}$.  By construction $\CC(x)=\{x\}$ for $x\in X$,  hence $\mathscr{C}_G$ satisfies \AX{(KS)}. For $v\in V(G)$, we write
	$\Anc(v) = \{w\in V(G) \mid v\preceq w\}$ for the ancestors of $v$.  For every
	leaf $x\in X$ we have $\Anc(x)=\{v \mid x\in \CC(v)\}$.  We write
	$\Anc(Y)\coloneqq \bigcap_{w\in Y} \Anc(w)$ for the set of common ancestors
	of all $w\in Y$.  In general, not every set $Y\subseteq V$ has a common
	ancestor in a DAG: Consider the DAG with three leaves $\{x,y,z\}$, two
	maximal vertices $\{p,q\}$, $\CC(p)=\{x,y\}$, and $\CC(q)=\{x,z\}$.  Then
	$\Anc(\{x,z\})=\emptyset$. A \emph{(rooted) network} $G$ is a DAG such that
	there is a unique vertex $\rho\in V(G)$, called the root, with indegree
	0. In a network, we have $x\preceq \rho$ for all $x\in V(G)$ and, thus, in
	particular, $\CC(\rho)=X$, i.e., $X\in \mathscr{C}_G$, and thus
	$\mathscr{C}_G$, satisfying \AX{(K1)}, is a clustering system.
	
	\begin{definition} \cite{Bender:01}\
		\label{def:lca}
		If $G$ is a DAG then $w$ is a \emph{least common ancestor} (LCA) of
		$Y\subseteq V(G)$ if it is a $\preceq$-minimal element in $\Anc(Y)$.  The
		set $\LCA(Y)$ comprises all LCAs of $Y$ in $G$.
	\end{definition}
	An LCA of $Y$ thus is an ancestor of all vertices in $Y$ that is
	$\preceq$-minimal w.r.t.\ this property. Clearly, $\LCA(\{v\})=\{v\}$ for
	all $v\in V(G)$ and $\LCA(Y)=\emptyset$ if and only if $\Anc(Y)=\emptyset$.
	In a network, the root vertex is a common ancestor for any set of vertices,
	and thus $\LCA(Y)\ne\emptyset$.
	
	We will, in particular, be interested in situations where the LCA of certain
	sets of leaves is uniquely defined. More precisely, we are interested in
	DAGs where $|\LCA(Y)|=1$ holds for certain subsets $Y\subseteq X$; the most
	obvious examples are DAGs that satisfy the $\2lca$-property (also known as
	the \emph{pairwise lca-property} \cite{Hellmuth:22q}), i.e., for every pair
	of leaves $x,y\in L(G)$ there is a unique least common ancestor
	$\lca(x,y)$. For simplicity, we will write $\lca(Y)=q$ instead of
	$\LCA(Y)=\{q\}$ whenever $|\LCA(Y)|=1$ and say that \emph{$\lca(Y)$ is
		defined}; otherwise, we leave $\lca(Y)$ \emph{undefined}.
	
	The following result for networks \cite[L.\ 17]{Hellmuth:22q} remains
	valid for all DAGs.
	\begin{lemma}
		\label{L3.33}
		Let $G$ be a DAG. Then $v\preceq_G w$ implies $\CC(v)\subseteq\CC(w)$ for
		all $v,w\in V(G)$.
	\end{lemma}
	Consequently, \cite[Obs.\ 12 \& 13]{Hellmuth:22q} also hold for DAGs in
	general:
	\begin{fact} 
		\label{obs:deflcaY}
		Let $G$ be a DAG with leaf set $X$, $\emptyset\ne A\subseteq X$, and
		suppose $\lca(A)$ is defined. Then the following is satisfied:
		\begin{enumerate}[noitemsep,nolistsep]
			\item[(i)] $\lca(A)\preceq_{G} v$ for all $v$ with $A\subseteq\CC(v)$.
			\item[(ii)] $\CC(\lca(A))$ is the unique inclusion-minimal cluster in
			$\mathscr{C}_G$ containing $A$.
			\item[(iii)] $\lca(\CC(\lca(A)))=\lca(A)$.
		\end{enumerate}
	\end{fact}
	Note that the existence of $\lca(A)$ for all $A\subseteq X$ does not
	imply that $G$ is a network since we could expand any network ``upward''
	for $\rho$ by attaching an arbitrary DAG that has $\rho$ as its unique
	leaf. Clearly, the vertices ``above'' $\rho$ cannot be least common
	ancestors of any leaves.
	
	Consider a set system $\mathscr{Q}\subseteq 2^X$. Then the Hasse diagram
	$\Hasse(\mathscr{Q})$ is the DAG with vertex set $\mathscr{Q}$ and directed
	edges from $A\in\mathscr{Q}$ to $B\in\mathscr{Q}$ if (i) $B\subsetneq A$
	and (ii) there is no $C\in\mathscr{Q}$ with $B\subsetneq C\subsetneq A$. As
	we shall see later, Hasse diagrams are of interest here because they
	guarantee well-behaved least common ancestors.
	
	The correspondence between Hasse diagrams that are networks and $k$-ary
	transit function is summarized in the following
	\begin{lemma}\label{lem:Hasse-network}
		Let $R$ be a $k$-ary transit function. Then $\Hasse(\mathscr{C}_R)$ is a
		network if and only if $R$ satisfies \AX{(a')} for $k$.
	\end{lemma}
	\begin{proof}
		If $N\coloneqq\Hasse(\mathscr{C}_R)$ is a network, it contains a unique
		vertex $\rho$ with indegree $0$; the root of $N$. Since $R$ satisfies
		\AX{(t3)}, all singletons $\{x\}$ with $x\in X$ are contained as vertices
		of $N$. Since $N$ has a unique root, it follows that $X$ is a vertex of
		$N$ and, in particular, $\CC(\rho)=X \in \mathscr{C}_R$.  This implies
		that there must be a subset $U\in X^{(k)}$ such that $R(U)=X$.  Hence,
		$R$ satisfies \AX{(a')}.  Conversely, if \AX{(a')} with parameter $k$
		holds, there is a subset $U\in X^{(k)} $ with $R(U)=X$ and thus
		$X\in\mathscr{C}_R$.  Let $v_X$ be the vertex in $\Hasse(\mathscr{C}_R)$
		for which $\CC(v_X)=X$ holds.  Since $v\preceq v_{X}$ for every vertex
		$v$ in $\Hasse(\mathscr{C}_R)$, this is in particular true for the
		singletons, and thus $v_X$ serves as the unique root of
		$\Hasse(\mathscr{C}_R)$.  \qed
	\end{proof}
	
	Following \cite{Hellmuth:22q}, we say that a DAG $G=(V,E)$ has the
	\emph{path-cluster-comparability \AX{(PCC)}} property if it satisfies, for
	all $u, v\in V$: $u$ and $v$ are $\preceq_G$-comparable if and only if
	$\CC(u)\subseteq \CC(v)$ or $\CC(v)\subseteq \CC(u)$.  By \cite[Cor.\ 11 \&
	Prop.\ 3]{Hellmuth:22q}, the Hasse diagram $G$ of a clustering system 
	$\mathscr{C}$ satisfies \AX{(PCC)} and \cite[Prop.\ 2]{Hellmuth:22q}
	implies that $\mathscr{C}_G=\mathscr{C}$.
	
	\section{DAGs with $\boldsymbol{\lca}$- and $\boldsymbol{\klca}$-property}
	
	In the following, we consider the generalization of lca-networks introduced
	in \cite{Hellmuth:22q} for arbitrary (not necessarily rooted) DAGs.
	\begin{definition} 
		A DAG with leaf set $X$ has the $\lca$-property if $\lca(A)$ is defined
		for all non-empty $A\subseteq X$.
	\end{definition}
	By definition, every DAG with the lca-property also has the pairwise
	lca-property.  The converse is, in general, not satisfied. An
	example of a network (rooted DAG) that satisfies the pairwise lca-property
	but that is not an lca-network, can be found in
	\cite[Fig.\ 13(A)]{Hellmuth:22q}.
	
	\begin{lemma}\label{lem:lca->closed}
		If a DAG $G$ has the $\lca$-property then its clustering system
		$\mathscr{C}_G$ is closed.
	\end{lemma}
	\begin{proof}
		To show that $\mathscr{C}_G$ is closed, we use the equivalent condition
		that $\mathscr{C}_G$ is closed under pairwise intersection.  Thus, let
		$\CC(u), \CC(v) \in \mathscr{C}_G$ for some $u,v\in V(G)$.  If $\CC(u)
		\subseteq \CC(v)$, $\CC(v) \subseteq \CC(u)$ or $\CC(u) \cap
		\CC(v)=\emptyset$, there is nothing to show. Hence, assume that
		$\CC(u)\overlaps \CC(v)$ and set $A\coloneqq\CC(u)\cap
		\CC(v)\ne\emptyset$. Since $G$ has the $\lca$-property, there is $w\in
		V(G)$ such that $w=\lca(A)$, and thus $A\subseteq \CC(w)$.
		The contraposition of Lemma~\ref{L3.33} shows that $u$ and $v$ are two
		incomparable common ancestors of $A$. Since $w$ is the unique
		$\preceq$-minimal common ancestor of $A$, we have $w\preceq u$ and
		$w\preceq v$, which -- together with Lemma \ref{L3.33} -- implies
		$\CC(w)\subseteq \CC(u)$ and $\CC(w)\subseteq \CC(v)$. Therefore
		$\CC(w)\subseteq A$.  Hence $A=\CC(w)\in\mathscr{C}_G$ and thus,
		$\mathscr{C}_G$ is closed.
		\qed
	\end{proof}
	The converse of Lemma~\ref{lem:lca->closed} is not true. A counter-example
	can be found in Fig.~\ref{fig:counter-ex-kc}.  The following connection
	between the clusters, the least common ancestors, and the closure function
	will be useful in the remainder of this contribution:
	\begin{fact}\label{L6.22}
		If $G$ is a DAG with leaf set $X$ and $\lca$-property, then
		$\CC(\lca(Y))=\cl(Y)$ for all $\emptyset \neq Y\subseteq X$.
	\end{fact}
	\begin{proof}
		The argument follows the proof of \cite[L.\ 41]{Hellmuth:22q}, observing
		that \cite[L.17]{Hellmuth:22q} remains true for arbitrary DAGs, and
		substituting Obs.~\ref{obs:deflcaY}(i) and Lemma~\ref{lem:lca->closed}
		for \cite[Cor.\ 18]{Hellmuth:22q} and \cite[P.\ 11]{Hellmuth:22q},
		respectively.
		\qed
	\end{proof}

	\begin{figure}[tb]
		\begin{minipage}{0.30\textwidth}
			\includegraphics[width=1.\textwidth]{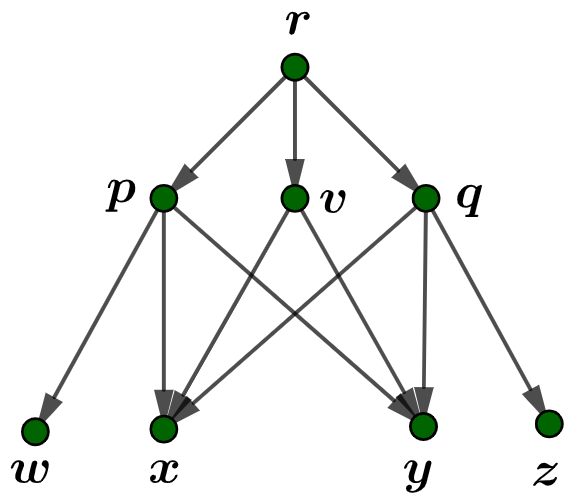}
		\end{minipage} \quad 
		\begin{minipage}{0.67\textwidth}
			\caption{The cluster system
				$\mathscr{C}_G=\big\{\{w\},\{x\},\{y\},\{x\},\{x,y\},\{w,x,y\},
				\{x,y,z\},\{w,x,y,z\}\big\}$ of the network $G$ is closed and
				satisfies \AX{(KC)} for every $k\in \{1,2,3,4\}$. However, we have
				$\LCA(\{x,y\})=\{p,v,q\}$ and thus $G$ does not have the pairwise
				$\lca$-property. By Prop.~\ref{prop:klca-implications}, if $G$ has
				the pairwise lca-property, then $\mathscr{C}_G$ is pre-binary. The
				example in this Figure shows that the converse is not true.
				In particular, the equivalence between pre-$k$-ary and the
				$\klca$-property in Prop.~\ref{prop:pLCA-pre-k-ary} requires
				\AX{(PCC)}, which is not satisfied by $G$.}
			\label{fig:counter-ex-kc}
		\end{minipage}
	\end{figure}

	The observations above can be extended to networks where more least common
	ancestors exist and are unique for all leaf sets of size at most $k$.
	Naturally, we start from the cluster system $\mathscr{C}\coloneqq
	\mathscr{C}_G\coloneqq\{\CC(v); v\in V(G)\}$ and consider the
	map $R_{\mathscr{C}}:X^k\to 2^X$ defined by
	\begin{equation*}
		R_{\mathscr{C}}(u_1,u_2,\dots,u_k) \coloneqq
		\bigcap\{ \CC(v) \mid  v\in V(G), u_1,u_2,\dots,u_k\in \CC(v) \} 
	\end{equation*}
	One easily verifies that $R_{\mathscr{C}}$ satisfies \AX{(t2)} and, thus, we can
	again interpret $R_{\mathscr{C}}$ as a function over sets in which case we
	have $R_{\mathscr{C}}(U)=\cl(U)$ for all $U\in X^{(k)}$.  In this setting,
	we are interested in cases where $\lca(U)$ is defined at least for all sets
	of cardinality $|U|\le k$. We formalize this idea in
	\begin{definition} 
		A DAG $G$ with leaf set $X$ has the $\klca$-property if $\lca(A)$ is
		defined for all $A\in X^{(k)}$.
	\end{definition}
	
	Now, we define the $k$-ary map $R_G: X^k\to
	2^X$ by $R_G(u_1,\dots,u_k) \coloneqq \CC(\lca(u_1,\dots,u_k))$; in set
	notation this reads $R_G(U)=\CC(\lca(U))$ for all $U\in X^{(k)}$.
	\begin{proposition}\label{prop:klca-implications}
		Let $G$ be a DAG with $\klca$-property. Then $R_G$ is a monotone, $k$-ary
		transit function that satisfies $R_G= R_{\mathscr{C}_G}$. Moreover, 
		$\mathscr{C}_G$ is pre-$k$-ary.  
	\end{proposition}
	\begin{proof}
		Let $G$ be a DAG with $\klca$-property and leaf set $X$.  It follows
		directly from the definition and uniqueness of $\lca(U)$ for $U\in
		X^{(k)}$ that $R_G$ satisfies \AX{(t1)}, \AX{(t2)} and \AX{(t3)}, i.e,
		$R_G$ is a $k$-ary transit function. If $u_1,\dots, u_k\in
		R_G(x_1,\dots,x_k)= \CC(\lca(x_1,\dots,x_k))$, then $\{u_1,\dots,
		u_k\}\subseteq \CC(\lca(x_1,\dots,x_k))$ and we can apply
		Obs.\ref{obs:deflcaY}(i) to conclude that $\lca(u_1,\dots, u_k)\preceq
		\lca(x_1,\dots ,x_k)$. Applying Lemma \ref{L3.33} yields
		$\CC(\lca(u_1,\dots, u_k)) \subseteq \CC(\lca(x_1,\dots ,x_k))$, and thus
		$R_G(u_1,\dots, u_k)\subseteq R_G(x_1,\dots,x_k)$, i.e., $R_G$ is
		monotone.  It follows from Obs.~\ref{obs:deflcaY}(ii) that $\CC(\lca(U))$
		is the unique inclusion minimal cluster in $\mathscr{C}_G$ containing
		$U$, i.e., $\cl(U)=\CC(\lca(U))$ for all $U\in X^{(k)}$.  Consequently,
		$R_G=R_{\mathscr{C}_G}$.
		
		Since $G$ has the $\klca$-property, $\lca(U)$ is
		defined for all $U\in X^{(k)}$. Thus $\CC(\lca(U))$ is the unique
		inclusion minimal cluster in $\mathscr{C}_G$ containing $U$ for all 
		$U\in X^{(k)}$ by Obs.~\ref{obs:deflcaY}(ii); hence $\mathscr{C}_{G}$ satisfies
		\AX{(KC)} for $k$, i.e., $\mathscr{C}_G$ is pre-$k$-ary.
		\qed
	\end{proof}
	
	Note that a DAG $G$ for which $\mathscr{C}_G$ is pre-$k$-ary 
	does not necessarily have the $\klca$-property, see
	Fig.~\ref{fig:kc-lca} and Fig.~\ref{fig:counter-ex-kc} for a
	counter-example.  Moreover, \AX{(KR)} with parameter $k$ is not necessarily
	satisfied since the $\klca$-property does not claim the existence of
	clusters that are not associated with least common ancestors of a set $U\in
	X^{(k)}$. Therefore, $R_G$ need not identify $\mathscr{C}_G$.
	At least for an important subclass of DAGs there is a
	simple correspondence between the uniqueness of LCAs and a property of the
	$\mathscr{T}$-system.

	\begin{figure}[t]
		\begin{minipage}{0.25\textwidth}
			\includegraphics[width=\textwidth]{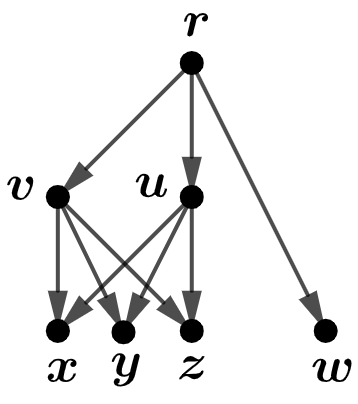}
		\end{minipage} \quad\qquad
		\begin{minipage}{0.65\textwidth}
			\caption{Consider the DAG $G$ with leaf set $X=L(G)$ where
				$\mathscr{C}_G=\{\{x\},\{y\},\{z\},\{w\},\{x,y,z\},X\}$.  Here,
				$\mathscr{C}_G$ satisfies \AX{(KS)} and \AX{(KC)} for $k=2$.  By
				definition, $\mathscr{C}_G$ is thus pre-binary. However, $G$ is not a
				pairwise $\lca$-network since $\lca(x,y)$, $\lca(x,z)$, and
				$\lca(y,z)$ are not defined. Moreover, $\mathscr{C}_G$ also satisfies
				\AX{(KC)} for $k=3$ but $G$ is not a \textit{3}-$\lca$-network since
				$\lca(x,y,z)$ is not defined.}
			\label{fig:kc-lca}
		\end{minipage}
	\end{figure}
	
	\begin{proposition}\label{prop:pLCA-pre-k-ary}
		Let $G$ be a DAG that satisfies \AX{(PCC)}. Then, $G$ satisfies the
		$\klca$-property if and only if $\mathscr{C}_G$ is pre-$k$-ary.
	\end{proposition}
	\begin{proof}
		Suppose that $G$ is a DAG with leaf set $X$ that satisfies \AX{(PCC)}. If $G$
		satisfies the $\klca$-property, then Prop.\ \ref{prop:klca-implications}
		implies that $\mathscr{C}_G$ is pre-$k$-ary.
		Assume now that $\mathscr{C}_G$ is pre-$k$-ary.  Hence, for all $U\in
		X^{(k)}$ we have $R_{\mathscr{C}_G}(U)=\bigcap\{C\in\mathscr{C}_G\mid
		U\subseteq C\} \in \mathscr{C}_G$.  Therefore, $\Anc(\{U\})\neq
		\emptyset$ and thus, $\LCA(\{U\})\neq \emptyset$. Assume, for
		contradiction, that there are two distinct vertices $v,w\in
		\LCA(U)$. Note that $U\subseteq R_{\mathscr{C}_G}(U)\subseteq \CC(v)\cap
		\CC(w)$.  By Def.\ \ref{def:lca}, both $v$ and $w$ are $\prec$-minimal
		ancestors of the vertices in $U$ and, therefore, $v$ and $w$ are
		incomparable in $G$. This, together with the fact that $G$ satisfies \AX{(PCC)}
		implies that neither $\CC(v)\subseteq \CC(w)$ nor $\CC(w)\subseteq
		\CC(v)$ can hold.  This and $\CC(v)\cap \CC(w)\neq \emptyset$ implies
		that $\CC(v)\overlaps\CC(w)$.  Since $\mathscr{C}_G$ is pre-$k$-ary,
		$R_{\mathscr{C}_G}(U)\in \mathscr{C}_G$, i.e., there is a vertex $z\in
		V(G)$ such that $\CC(z) = R_{\mathscr{C}_G}(U)$. Hence, $\CC(z)\subseteq
		\CC(w)\cap \CC(v)$. Since $\CC(v)\overlaps\CC(w)$ it must hold that
		$\CC(z)\subsetneq \CC(w)$ and $\CC(z)\subsetneq \CC(v)$.  Since $G$
		satisfies \AX{(PCC)}, $z$ and $v$ must be $\preceq$-comparable.  If, however,
		$v\preceq z$, then Lemma \ref{L3.33} implies that $\CC(v)\subseteq
		\CC(z)$; a contradiction. Hence, $z\prec v$ and, by similar arguments,
		$z\prec w$ must hold.  This, however, contradicts the fact that $v$ and
		$w$ are $\prec$-minimal ancestors of all the vertices in $U$. Hence,
		$|\LCA(U)|=1$ must hold for all $U\in X^{(k)}$.  Consequently, $G$
		satisfies the $\klca$-property.  \qed
	\end{proof}
	
	Consequently, we obtain a characterization of pre-$k$-ary clustering
	systems in terms of the DAGs from which they derive.
	\begin{theorem}\label{thm:char-pre-k-ary}
		A clustering system $\mathscr{C}$ is pre-$k$-ary if and only if there is
		a DAG $G$ with $\mathscr{C} = \mathscr{C}_G$ and $\klca$-property.
	\end{theorem}
	\begin{proof}
		Suppose that $\mathscr{C}$ is a pre-$k$-ary clustering system and
		consider the Hasse diagram $G\coloneqq \Hasse(\mathscr{C})$. It satisfies
		\AX{(PCC)} and $\mathscr{C}_G=\mathscr{C}$.  Consequently,
		$\mathscr{C}_G$ is pre-$k$-ary. Thus, we can apply
		Prop.~\ref{prop:pLCA-pre-k-ary} to conclude that $G$ satisfies the
		$\klca$-property.
		Conversely, suppose that $G$ is a DAG with the $\klca$-property and
		$\mathscr{C} = \mathscr{C}_G$. By Prop.\ \ref{prop:klca-implications},
		$\mathscr{C}$ is pre-$k$-ary.  \qed
	\end{proof}
	
	Next, we show that $k$-ary transit functions give rise to DAGs
	with the $\klca$-property in a rather natural way:
	\begin{lemma}\label{lem:Hasse-k}
		Let $R$ be a monotone $k$-ary transit function. Then, the Hasse diagram of
		its transit sets $\Hasse(\mathscr{C}_R)$ satisfies the $\klca$-property.
	\end{lemma}
	\begin{proof}
		Let $R$ be a monotone transit function on $X$ and $U\in X^{(k)}$.
		Considering $R$ as a function over subsets, conditions \AX{(t1)} and
		\AX{(t2)} imply that $U\subseteq R(U)$.  In the following, let $v_C$
		denote the unique vertex in $\Hasse(\mathscr{C}_R)$ that corresponds to
		the cluster $C\in \mathscr{C}_R$.  For all $W\in X^{(k)}$ with
		$U\subseteq R(W)$ it holds, by condition \AX{(m)}, that $R(U) \subseteq
		R(W)$. This, together with the definition of the Hasse diagram implies
		that $v_{R(U)}\preceq v_{R(W)}$ for all $W\in X^{(k)}$ with $U\subseteq
		R(W)$.  Thus, $v_{R(U)}$ is the unique $\preceq$-minimal vertex in
		$\Hasse(\mathscr{C}_R)$ satisfying $x\preceq v_{R(U)}$ for all $x\in U$,
		and thus $v_{R(U)}=\lca(U)$.  \qed
	\end{proof}
	As an immediate consequence of the correspondence between monotone
	$k$-ary transit functions and $k$-ary $\mathscr{T}$-systems, we also
	conclude that the Hasse diagram of 
	$k$-ary $\mathscr{T}$-systems has
	the $\klca$-property.
	
	The converse of Lemma~\ref{lem:Hasse-k}, however, need not be true: A Hasse
	diagram $\Hasse(\mathscr{C}_R)$ with the $\klca$-property for some $k$ is
	not sufficient to imply that $R$ is monotone:
	
	\begin{example}
		\label{ex:nonmonotone}  
		Let $R$ on $X=\{a,b,c,d\}$ be symmetric and defined by $R(a,b)=X$,
		$R(a,c)=\{a,b,c\}$ and all other sets are singletons or $X$ in such a way
		that \AX{(t1)} and \AX{(t3)} is satisfied.  One easily verifies that $R$
		is a transit function satisfying \AX{(a')} and that
		$\Hasse(\mathscr{C}_R)$ is a network with root $X$. In fact,
		$\Hasse(\mathscr{C}_R)$ is a rooted tree having pairwise
		lca-property. However, $R$ is not monotone since $R(a,b)=X\nsubseteq
		R(a,c)$.
	\end{example}
	\begin{theorem}\label{thm:R-monotone-char}
		Let $R$ be a $k$-ary transit function. Then $R$ is monotone if and only
		if there is a DAG $G$ with $\klca$-property and that satisfies
		$\mathscr{C}_G = \mathscr{C}_R$ and $R_{\mathscr{C}_G}=R$.
	\end{theorem}
	\begin{proof}
		If $R$ is a monotone $k$-ary transit function, then
		$G=\Hasse(\mathscr{C}_R)$ satisfies $\mathscr{C}_G = \mathscr{C}_R$.  By
		Lemma \ref{lem:Hasse-k}, $G$ has the $\klca$-property.  Moreover, since
		$\mathscr{C}_G = \mathscr{C}_R$ it follows that
		$R_{\mathscr{C}_G}=R_{\mathscr{C}_R}$. Since $R$ is monotone,
		$R =R_{\mathscr{C}_R}=R_{\mathscr{C}_G}$.
		
		Conversely, let $G$ be a DAG with $\mathscr{C}_G = \mathscr{C}_R$ and
		$\klca$-property. By Prop.~\ref{prop:klca-implications}, $R_G
		= R_{\mathscr{C}_G}$ is monotone. Therefore, $R$ is a monotone  $k$-ary transit function.
	\end{proof}
	
	\section{DAGs with strict and strong  $\boldsymbol{\klca}$-property.}
	
	In general, $\lca(\CC(w))$ is not necessarily defined for all $w\in V(G)$,
	see e.g.\ the DAG in Fig.~\ref{fig:kc-lca}. As discussed in
	\cite{Hellmuth:22q}, it is, however, a desirable property:
	\begin{description}[noitemsep,nolistsep] 
		\item[\AX{(CL)}] For every $v\in V(G)$, $\lca(\CC(v))$ is defined.
	\end{description}
	By definition, every DAG $G$ that has the lca-property satisfies \AX{(CL)}.
	\begin{lemma}\label{L6.9}
		Let $G$ be a DAG satisfying \AX{(CL)}. Then $\CC(\lca(\CC(v)))=\CC(v)$
		for all $v\in V(G)$.
	\end{lemma}
	\begin{proof}
		Since $\lca(\CC(v))$ is defined, Obs.~\ref{obs:deflcaY}(i) implies
		$\lca(\CC(v))\preceq v$. Then by Lemma~\ref{L3.33},
		$\CC(\lca(\CC(v)))\subseteq\CC(v)$. The reverse inclusion is trivial.
		\qed
	\end{proof}
	
	\begin{definition}
		Let $G$ be a DAG with leaf set $X$ and $\klca$ property. Then, $G$ has
		the \emph{strict $\klca$-property} if $G$ satisfies \AX{(CL)} and for
		every $w\in V(G)$ there is $U\in X^{(k)}$ such that
		$\lca(\CC(w))=\lca(U)$.
	\end{definition}
	
	\begin{proposition}\label{prop:klca-strict-kkr}
		Let $G$ be a DAG with leaf set $X$ and $\klca$-property. Then $G$ has the
		strict $\klca$-property if and only if $\mathscr{C}_G$ is a $k$-ary $\mathscr{T}$-system.
		In this case, $\mathscr{C}_G$ is identified by $R_G$.
	\end{proposition}
	\begin{proof}
		The $\klca$-property of $G$ implies that $\mathscr{C}_G$ is pre-$k$-ary, from Prop.~\ref{prop:klca-implications}.
		First assume that $G$ has the strict $\klca$-property. 
		Consider $\CC(w)\in \mathscr{C}_G$.  By
		definition, there exists $U\in X^{(k)}$ such that
		$\lca(\CC(w))=\lca(U)$. Since $G$ satisfies \AX{(CL)}, Lemma~\ref{L6.9}
		implies $\CC(w)=\CC(\lca(\CC(w)))=\CC(\lca(U))$.  Moreover, by
		Obs.~\ref{obs:deflcaY}(ii), $\CC(\lca(U))=\CC(w)$ is the unique inclusion
		minimal cluster in $\mathscr{C}_G$ containing $U$. This implies both
		$U\subseteq \CC(w)$ and $C(w)\subseteq C(v)$ for every $v\in V(G)$ with
		$U\subseteq C(v)$. Hence, $\mathscr{C}_G$ satisfies \AX{(KR)}. Hence, $\mathscr{C}_G$ is a $k$-ary $\mathscr{T}$-system.
		
		Conversely, assume that $G$ holds $\klca$-property and $\mathscr{C}_G$
		satisfies \AX{(KR)}. Thus, for every $w\in V(G)$, there is $U\in X^{(k)}$
		such that $U\subseteq C(w)$ and $U\subseteq C(v)$ implies $C(w)\subseteq
		C(v)$ for all $v\in V(G)$. Hence, $C(w)$ is an inclusion minimal set in
		$\mathscr{C}_G$ containing $U$.  Since $G$ has the $\klca$-property,
		$\lca(U)$ is defined and, by Obs.~\ref{obs:deflcaY}(ii), $\CC(\lca(U))$
		is the unique inclusion minimal set in $\mathscr{C}_G$ containing $U$,
		and thus $\CC(w)=\CC(\lca(U))$ must hold. By Obs.~\ref{obs:deflcaY}(iii)
		we have $\lca(\CC(w))=\lca(\CC(\lca(U)))=\lca(U)$. Therefore, $G$ has the
		strict $\klca$-property.
		
		Since a set system is identified by a $k$-ary transit function if and
		only if it is a $k$-ary $\mathscr{T}$-system and its canonical transit function identifies it, we have $\mathscr{C}_G$ is identified by $R_{\mathscr{C}_G}$. Moreover, $R_{\mathscr{C}_G}=R_G$ from Prop.~\ref{prop:klca-implications}. Hence the result.
		\qed
	\end{proof}
	
	In \cite{Hellmuth:22q}, networks with the \emph{strong $\lca$-property}
	were introduced. These satisfy (i) the lca-property and (ii) for every
	non-empty subset $A\subseteq X$, there are $x,y\in A$ such that
	$\lca(x,y)=\lca(A)$. As it turns out, these networks are characterized by
	their clustering systems: $G$ is a strong $\lca$-network if and only if $G$  has the $\lca$-property and $\mathscr{C}_G$ is a weak hierarchy
	\cite[Prop.\ 13]{Hellmuth:22q}. In the following, we generalize these
	results to DAGs in general and spanning sets for $\lca(A)$ that are larger
	than a pair of points:
	\begin{definition}
		Let $G$ be DAG with leaf set $X$ and $\lca$-property. Then, $G$ has the
		\emph{strong $\klca$-property} if, for every non-empty subset $A\subseteq
		X$, there is $U\in X^{(k)}$ such that $\lca(U)=\lca(A)$.
	\end{definition}
	Fig.~\ref{fig:lca-wh} shows that the $\lca$ property does not imply
	the strong $\klca$-property, i.e., the uniqueness of LCAs for all
	$A\subseteq X$ does not imply that these are spanned by small subsets of leaves.
	\begin{lemma}
		If a DAG $G$ has the strong $\klca$-property, then it has the strict $\klca$-property.
	\end{lemma}
	\begin{proof}
		Suppose that $G$ is a DAG with leaf set $X$ and that has the strong
		$\klca$-property.  By definition, $G$ has the $\lca$-property. Hence, for
		all non-empty $A\in 2^X$, $\lca(A)$ is defined. This implies
		that $\lca(A)$ is defined for all $A\in X^{(k)}\subseteq 2^X$ and thus,
		$G$ has the $\klca$-property.  Furthermore, since $\CC(v)\in 2^X$ for all
		$v\in V(G)$, the DAG $G$ satisfies \AX{(CL)}. Let $w\in V(G)$. Since
		$A\coloneqq \CC(w)\subseteq X$ and since $G$ has the strong
		$\klca$-property, there exists $U\in X^{(k)}$ such that
		$\lca(U)=\lca(A)=\lca(\CC(w))$. In summary, $G$ has the strict
		$\klca$-property.\qed
	\end{proof}
	\begin{figure}[tb]
		\begin{minipage}{0.25\textwidth}
			\includegraphics[width=\textwidth]{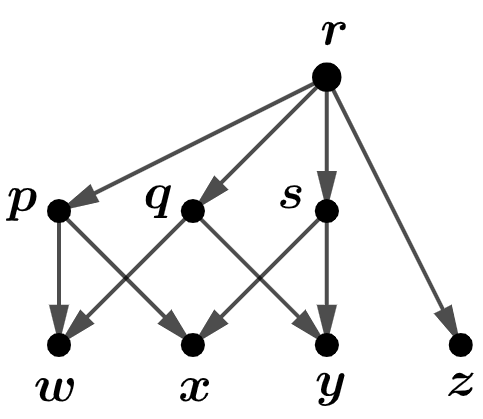}
		\end{minipage} \quad\qquad
		\begin{minipage}{0.65\textwidth}
			\caption{The network $G$ with leaf set $X=\{w,x,y,z\}$ has the
				clustering systems
				$\mathscr{C}_G=\{\{x\},\{y\},\{z\},\{w\},\{w,x\},\{w,y\},
				\{x,y\},X\}$ satisfying \AX{(KS)}, \AX{(KC)} and \AX{(KR)} for
				$k =2,3,4$. Moreover, $\lca(u,v)$ is defined for all $u,v\in
				X$. However, the sets $\{w,x\},\{w,y\}, \{x,y\}$ violates the
				condition of weak hierarchy. $G$ is an $\lca$-network but not a
				strong-2-$\lca$-network since $\lca(\{w,x,y\})=r\neq
				\lca(\{u,v\})$ for any $u,v\in \{w,x,y\}$.}
			\label{fig:lca-wh}
		\end{minipage}
	\end{figure}
	\begin{proposition}\label{prop:strong-klca}
		Let $G$ be a DAG with leaf set $X$ and $\lca$-property. Then $G$ has the
		strong $\klca$-property if and only if for every non-empty subset
		$A\subseteq X$ there exists $U\subseteq A$ with $|U|\leq k$ such that
		$\cl(A)=\cl(U)$ in $\mathscr{C}_G$.
	\end{proposition}
	\begin{proof}
		Assume first that $G$ has the strong $\klca$-property and let
		$\emptyset\ne A\subseteq X$. Then $\lca(A)=\lca(U)$ for some $U\subseteq
		X$ with $|U|\leq k$. Applying Obs.~\ref{L6.22}, we obtain
		$\cl(A)=\CC(\lca(A))=\CC(\lca(U))=\cl(U)$.  Conversely, assume that for
		every non-empty subset $A\subseteq X$, there exists $U\subseteq A$ with
		$|U|\leq k$ such that $\cl(A)=\cl(U)$ in $\mathscr{C}_G$.  Let
		$A\subseteq X$ be non-empty. Applying Obs.~\ref{obs:deflcaY}(iii) and
		Obs.~\ref{L6.22}  yields
		$\lca(A)=\lca(\CC(\lca(A)))=\lca(\cl(A))=\lca(\cl(U))=
		\lca(\CC(\lca(U)))=\lca(U)$.  \qed
	\end{proof}
	
	Prop.~\ref{prop:strong-klca} and \ref{prop:kwH} imply
	\begin{theorem}
		\label{thm:main}
		$G$ is DAG with the strong $\klca$-property if and only if $G$ has the
		$\lca$-property and $\mathscr{C}_G$ is a $k$-weak hierarchy.
	\end{theorem}
	
	\section{Concluding Remarks} 
	
	The connection between clusters and LCAs in DAGs is not limited to the
	relationships discussed so far and summarized in the following diagram:
	\begin{equation*}
		\includegraphics[width=0.7\textwidth]{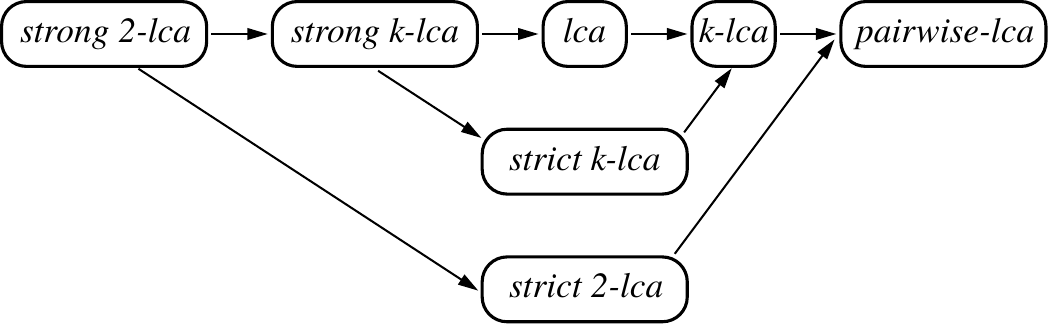}
	\end{equation*}
	Note that there are no implications between the $\lca$-, strict $\2lca$- 
	and strict $\klca$-property.
	
	\emph{Weak pyramids} are weak hierarchies that, in addition, satisfy a
	necessary (but not sufficient) condition for $\mathscr{C}$ to comprise
	intervals \cite{Nebesky:83} called \AX{(WP)} in \cite{Changat:21w}.  One
	can show, for instance, that $\mathscr{C}_G$ is weakly pyramidal for a DAG
	with the strong $2$-$\lca$-property if and only if no four distinct
	vertices $a,b,c,d\in X$ exist such that $b,c\npreceq \lca(a,d)$,
	$a,c\npreceq \lca(b,d)$, and $a,b\npreceq \lca(c,d)$. Results like this,
	which we can only mention here due to space restrictions, suggest that the
	connection between LCAs and clusters in DAGs remains an interesting topic for
	future research.
	
	\bigskip\textbf{Acknowledgments.} AVS acknowledges the financial support from the CSIR-HRDG for the Senior Research Fellowship(09/0102(12336)/2021-EMR-I).

	\bibliographystyle{splncs04}
	\bibliography{LCA}
\end{document}